\documentclass[11pt]{article}
 \usepackage{geometry}
\usepackage{verbatim}
\usepackage[rgb]{xcolor}
\usepackage{tikz}
\usetikzlibrary{shapes,snakes,calendar,matrix,backgrounds,folding}
\usepackage{graphicx}
\usepackage{epsfig}
\usepackage{alltt}
\usepackage{amsmath}
\usepackage{epsfig}
\usepackage{amssymb,amsfonts,epsf,url}
\usepackage{pstricks,pstricks-add}
\usepackage{pst-coil}
\usepackage{pst-node}
\usepackage{clrscode}

\usepackage{fancyhdr}

\newcommand{\important}[2] {\node[name=#1, minimum size=8mm] {#2};}
\newcommand{\leaf}[2] {\node[rectangle,name=#1, minimum size=5mm] {#2};}

\newlength{\alginputwidth}
\newlength{\algboxwidth}

\newsavebox{\algbox}
\newsavebox{\captionbox}
    {
        \setlength{\algboxwidth}{\columnwidth}
        \addtolength{\algboxwidth}{-\columnsep}
        \addtolength{\algboxwidth}{-1mm}
        \setlength{\alginputwidth}{\algboxwidth}
        \addtolength{\alginputwidth}{-1.7cm}
        \begin{figure}[tb]
            \vspace*{2mm}
            \centering
            \begin{lrbox}{\captionbox}
                \begin{minipage}[b]{\algboxwidth}
                    \centering
                    \caption{#1}
                    \label{#2}
                \end{minipage}
            \end{lrbox}
            \begin{lrbox}{\algbox}
                \begin{minipage}[b]{\algboxwidth}
                    \footnotesize
                    \vspace*{2mm}
    } 
    {
                    \vspace*{0.2mm}
               \end{minipage}
            \end{lrbox}
            \fbox{\usebox{\algbox}\hspace*{1mm}}
            \usebox{\captionbox}
            \vspace*{-4mm}
        \end{figure}
    }
\newsavebox{\algcodebox}
    {
        \begin{enumerate}
            \setlength{\itemsep}{2pt}
            \setlength{\parsep}{0pt}
            \setlength{\topsep}{0pt}
            \setlength{\parskip}{0pt}
            \setlength{\partopsep}{0pt}
    } 
    {\end{enumerate}}

\newcommand{\paramproblem}[4]{\noindent {\sc #1}
\\
{\bf Given:} #2\\
{\bf Parameter:} #3\\
{\bf Question:} #4}

\newtheorem{theorem}{Theorem}[section]
\newtheorem{lemma}[theorem]{Lemma}

\newtheorem{proposition}[theorem]{Proposition}
\newtheorem{observation}[theorem]{Observation}
\newtheorem{reduction}{Reduction Rule}[section]

\newtheorem{case}[theorem]{Case}
\newtheorem{subcase}[theorem]{SubCase}

\newtheorem{branchrule}[theorem]{BranchRule}

\newcommand{\blackslug}{\penalty 1000\hbox{
    \vrule height 8pt width .4pt\hskip -.4pt
    \vbox{\hrule width 8pt height .4pt\vskip -.4pt
          \vskip 8pt
      \vskip -.4pt\hrule width 8pt height .4pt}
    \hskip -3.9pt
    \vrule height 8pt width .4pt}}

\newenvironment{proof}{\noindent {\sc Proof.}$\;\;$\rm}{\qed}

\newcommand{\qed}{\hspace*{\fill}\blackslug}

\newtheorem{definition}{Definition}[section]
\begin{document}
\psset{unit=1pt}
\title{Algorithms for Cut Problems on Trees}

\author{
  Iyad Kanj
    \thanks{School of  Computing, DePaul University,
        243 S. Wabash Avenue, Chicago, IL 60604, USA.
        Email: \texttt{ikanj@cs.depaul.edu}.}
  \and
  Guohui Lin
    \thanks{Department of Computing Science, University of Alberta,
        Edmonton, Alberta T6G 2E8, Canada.
        Email: \texttt{\{guohui,weitian\}@ualberta.ca}.}
  \and
  Tian Liu
    \thanks{Key Laboratory of High Confidence Software Technologies, Ministry of Education,
        Institute of Software, School of Electronic Engineering and Computer Science,
        Peking University, Beijing 100871, China.
        Email: \texttt{lt@pku.edu.cn}.}
  \and
  Weitian Tong
  \footnotemark[2]
  \and
  Ge Xia
    \thanks{Department of Computer Science, Lafayette College,
        506 Acopian Engineering Center, Easton, PA 18042, USA.
        Email: \texttt{xiag@lafayette.edu}.}
  \and
  Jinhui Xu
    \thanks{Department of Computer Science and Engineering,
        State University of New York at Buffalo (SUNY Buffalo),
        338 Davis Hall, Buffalo, NY 14260, USA.
        Email: \texttt{jinhui@buffalo.edu}.}
  \and
  Boting Yang
    \thanks{Department of Computer Science, University of Regina,
        Regina, Saskatchewan S4S 0A2, Canada.
        Email: \texttt{boting@cs.uregina.ca}.}
   \and
  Fenghui Zhang
    \thanks{Google
	Kirkland, 747 6th Street South,
	Kirkland, WA 98033, USA.
	\texttt{fhzhang@gmail.com}.}
  \and
  Peng Zhang
    \thanks{School of Computer Science and Technology, Shandong University,
        Jinan 250101, China.
        Email: \texttt{algzhang@sdu.edu.cn}.}
  \and
  Binhai Zhu
    \thanks{Department of Computer Science, Montana State University,
        Bozeman, MT 59717, USA.
        Email: \texttt{bhz@cs.montana.edu}.}
}
\date{}
\maketitle

\begin{abstract}
We study the {\sc multicut on trees} and the {\sc generalized multiway Cut on trees} problems. For the {\sc multicut on trees} problem, we present a parameterized algorithm that runs in time
$O^{*}(\rho^k)$, where $\rho = \sqrt{\sqrt{2} + 1} \approx 1.555$ is the positive root of the polynomial $x^4-2x^2-1$.
This improves the current-best algorithm of Chen et al. that runs in time $O^{*}(1.619^k)$. For the {\sc generalized multiway cut on trees} problem, we show that this problem is solvable in polynomial time if the number of terminal sets is fixed; this answers an open question posed in a recent paper by Liu and Zhang. By reducing the {\sc generalized multiway cut on trees} problem to the {\sc multicut on trees} problem, our results give a parameterized algorithm that solves the {\sc generalized multiway cut on trees} problem in time $O^{*}(\rho^k)$, where $\rho = \sqrt{\sqrt{2} + 1} \approx 1.555$ time.
\end{abstract}

\section{Introduction}
\label{sec:intro}
Let $T$ be a tree. We consider the following problems:

\paramproblem{{\sc multicut on trees} ({\sc MCT})}{A tree $T$ and a set $R$ of pairs of vertices of $T$ called {\em terminals}: $R=\{(u_1, v_1), \ldots, (u_r, v_r)\}$}{$k$}{Is there a set of at most $k$ edges in $T$ whose removal disconnects each $u_i$ from $v_i$, for $i=1, \ldots, r$?} \\

For convenience, we will refer to a pair of terminals $(u_i, v_i) \in R$ by a {\em request}, and we will also say that $u_i$ {\em has a request to} $v_i$, and vice versa.

\paramproblem{{\sc generalized multiway cut on trees} ({\sc GMWCT})}{A tree $T$ and and a collection of vertex/terminal-sets $S_1, \ldots S_r$} {$k$}{Is there a set of at most $k$ edges in $T$ whose removal disconnects each pair of vertices in the same terminal set $S_i$, for $i=1, \ldots, r$?}

As the name indicates, the {\sc GMWCT} problem generalizes the well-known {\sc multiway cut on trees} problem in which there is only one set of terminals.

The {\sc MCT} problem has applications in networking~\cite{survey}. The problem is known to be NP-complete, and its optimization version is APX-complete and has an approximation ratio of 2~\cite{DBLP:journals/algorithmica/GargVY97}. Assuming the Unique Games Conjecture, the {\sc MCT} problem cannot be approximated to within $2-\epsilon$ \cite{KR08}. From the parameterized complexity perspective, Guo and Niedermeier~\cite{guo} showed that the {\sc MCT} problem is fixed-parameter tractable by giving an $O^{*}(2^k)$ time algorithm for the problem. (The asymptotic notation $O^*(f(k))$ suppresses any polynomial factor in the input length.) They also showed that {\sc MCT} has an exponential-size kernel. Bousquet, Daligault, Thomass\'{e}, and Yeo, improved the upper bound on the kernel size for {\sc MCT} to $O(k^6)$~\cite{thomase}, which was subsequently improved very recently by Chen et al.~\cite{multicut} to $O(k^3)$. Chen et al.~\cite{multicut} also gave a parameterized algorithm for the problem running in time $O^*(1.619^k)$.

The {\sc multiway cut on trees} problem (i.e., there is one set of terminals) was proved to be
solvable in polynomial time in~\cite{CR91,CB04}. Chopra and Rao \cite{CR91}
first gave a polynomial-time greedy algorithm for the problem. More recently, Costa and Billionnet~\cite{CB04} proved that {\sc multiway cut on trees} can be solved
in linear time by dynamic programming. Very recently, Liu and Zhang~\cite{LZ12} generalized the {\sc multiway cut on trees} problem from one set of terminals to allowing multiple terminal sets, which results in the {\sc GMWCT} defined above. They showed that the {\sc GMWCT} problem is fixed-parameter tractable
by reducing it to the {\sc MCT} problem~\cite{LZ12}.  Clearly, the {\sc GMWCT} problem is NP-complete when the number of terminal sets is part of the input by a trivial reduction from the {\sc MCT} problem. Liu and Zhang asked about the complexity of the problem if the number of terminal sets is a constant (i.e., not part of the input)~\cite{LZ12}.

We mention that the {\sc multicut} and {\sc multiway cut} problems on general graphs are very important problems that have been extensively studied.
Marx~\cite{1140647} studied the parameterized complexity of several graph separation problems, including {\sc multicut} and {\sc multiway cut} on general graphs. Recently, the {\sc multicut} problem on general graphs was shown to be fixed-parameter tractable independently by Bousquet, Daligault, and Thomass\'{e}~\cite{newthomas}, and by Marx and Razgon~\cite{newmarx}, answering an outstanding open problem in parameterized complexity theory.
Very recently, Chitnis et al.~\cite{chitnis} proved that the {\sc multiway cut} problem on directed graphs is fixed-parameter tractable when parameterized by the size of the solution (i.e., cut set). Also very recently, Klein and Marx~\cite{kleinmarxplanar}, and Marx~\cite{marxplanar} gave upper bounds and lower bounds, respectively, on the parameterized complexity of the {\sc planar multiway cut} problem parameterized by the number of terminals.

In the current paper we present a parameterized algorithm that runs in time
$O^{*}(\rho^k)$, where $\rho = \sqrt{\sqrt{2} + 1} \approx 1.555$ is the positive root of the polynomial $x^4-2x^2-1$.
This improves the current-best algorithm of Chen et al.~\cite{multicut} that runs in time $O^{*}(1.619^k)$. This improvement is obtained by extending the connection between the {\sc MCT} problem and the {\sc Vertex Cover} problem, first exploited in the paper of Chen et al.~\cite{multicut}. For the {\sc GMWCT} problem, we show that the problem is solvable in polynomial time if the number of terminal sets is a constant; this answers the open question posed by Liu and Zhang. By reducing the {\sc GMWCT} problem to the {\sc MCT} problem, our result implies that the {\sc GMWCT} problem is also solvable in $O^{*}(\rho^k)$, where $\rho = \sqrt{\sqrt{2} + 1} \approx 1.555$ time.

\section{Preliminaries}
\label{sec:prelim}
We assume familiarity with basic graph theory and parameterized complexity notation and terminology. For more information, we refer the reader to~\cite{fptbook,grohe,rolf,west}.

For a graph $H$ we denote by $V(H)$ and $E(H)$ the set of vertices and edges of $H$, respectively. For a vertex $v \in H$, $H-v$ denotes $H[V(H) \setminus \{v\}]$, and for a subset of vertices $S \subseteq V(H)$, $H-S$ denotes $H[V(H) \setminus S]$. By {\em removing} a subgraph $H'$ of $H$ we mean removing $V(H')$ from $H$ to obtain $H - V(H')$. Two vertices $u$ and $v$ in $H$ are said to be {\em adjacent} or {\em neighbors} if $uv \in E(H)$. For two vertices $u, v \in V(H)$, we denote by $H - uv$ the graph $(V(H), E(H) \setminus \{uv\})$. By {\em removing} an edge $uv$ from $H$ we mean setting $H = H -uv$. For a subset of edges $E' \subseteq E(H)$, we denote by $H-E'$ the graph $(V(H), E(H) \setminus E')$. For a vertex $v \in H$, $N(v)$ denotes the set of neighbors of $v$ in $H$. The {\em degree} of a vertex $v$ in $H$, denoted $deg_H(v)$, is $|N(v)|$. The {\em degree} of $H$, denoted $\Delta(H)$, is $\Delta(H) = \max\{deg_H(v): v \in H\}$. The {\em length} of a path in a graph $H$ is the number of edges in it.
A {\em vertex cover} for a graph $H$ is a set of vertices such that each edge in $H$ is incident to at least one vertex in this set. A vertex cover for $H$ is {\em minimum} if its cardinality is minimum among all vertex covers of $H$; we denote by $\tau(H)$ the cardinality/size of a minimum vertex cover of $H$.

A {\em tree} is a connected acyclic graph. A {\em leaf} in a tree is a vertex of degree at most 1. A nonleaf vertex in a tree is called an {\em internal} vertex. For two vertices $u$ and $v$, the {\em distance} between $u$ and $v$ in $T$, denoted $dist_T(u, v)$, is the length of the unique path between $u$ and $v$ in $T$. A leaf $x$ in a tree is said to be {\em attached} to vertex $u$ if $u$ is the unique neighbor of $x$ in the tree. A {\em forest} is a collection of disjoint trees.

Let $T$ be a tree with root $r$. For a vertex $u \neq r$ in $V(T)$, we denote by $\pi(u)$ the parent of $u$ in $T$. A {\em sibling} of $u$ is a child $v \neq u$ of $\pi(u)$ (if exists), and an {\em uncle} of $u$ is a sibling of $\pi(u)$. A vertex $v$ is a {\em nephew} of a vertex $u$ if $u$ is an uncle of $v$. For a vertex $u \in V(T)$, $T_u$ denotes the subtree of $T$ rooted at $u$. The {\em children} of a vertex $u$ in $V(T)$ are the vertices in $N(u)$ if $u = r$, and in $N(u) - \pi(u)$ if $u \neq r$. A vertex $u$ is a {\em grandparent} of a vertex $v$ if $\pi(v)$ is a child of $u$. A vertex $v$ is a {\em grandchild} of a vertex $u$ if $u$ is a grandparent of $v$.

A {\it parameterized problem} is a set of instances of the form $(x,
k)$, where $x \in \Sigma^*$ for a finite alphabet set $\Sigma$, and
$k$ is a non-negative integer called the {\em parameter}.
A parameterized problem $Q$ is {\it fixed parameter tractable}, or
simply FPT, if there exists an algorithm that on input $(x, k)$
decides if $(x, k)$ is a yes-instance of $Q$ in time $f(k)|x|^{O(1)}$,
where $f$ is a computable function independent of $|x|$.

Let $(T, R)$ be an instance of {\sc multicut on trees}. A subset of edges $E' \subseteq E(T)$ is said to be an {\em edge cut}, or simply a {\em cut}, for $R$ if for every request $(u, v)$ in $R$, there is no path between $u$ and $v$ in $T - E'$. The {\em size} of a cut $E'$ is $|E'|$. A cut $E'$ is {\em minimum} if its cardinality is minimum among all cuts.

Let $(T, R, k)$ be an instance of {\sc multicut on trees}, and let $uv$ be an edge in $E(T)$. If we know that edge $uv$ can be included in the solution sought, then we can remove $uv$ from $T$ and decrement the parameter $k$ by 1; we say in this case that we {\em cut} edge $uv$. By {\em cutting} a leaf we mean cutting the unique edge incident to it. If $T$ is a rooted tree and $u \in T$ is not the root, we say that we {\em cut} $u$ to mean that we cut the edge $u\pi(u)$. On the other hand, if we know that edge $uv$ can be excluded from the solution sought, we say in this case that edge $uv$ is {\em kept}, and we can {\em contract} it by identifying the two vertices $u$ and $v$, i.e., removing $u$ and $v$ and creating a new vertex with neighbors $(N(u) \cup N(v)) \setminus \{u, v\}$). If edge $uv$ is contracted and $w$ is the new vertex, then any request in $R$ of the form $(u, x)$ or $(v, x)$ is replaced by the request $(w, x)$.

For a vertex $u$ in $T$, we define an auxiliary graph $G_u$ as follows. The vertices of $G_u$ are the leaves in $T$ attached to $u$ (if any). Two vertices $x$ and $y$ in $G_u$ are adjacent in $G_u$ if and only if there is a request between $x$ and $y$ in $R$. Without loss of generality, we shall call the vertices in $G$ with the same names as their corresponding leaves in $T$, and it will be clear from the context whether we are referring to the leaves in $T$ or to their corresponding vertices in $G_u$.

It is not difficult to see that if $C$ is a vertex cover for $G_u$ then the
edge-set $E_C = \{uw \in E(T) \mid w \in C\}$, which has the
same cardinality as $C$, cuts every request between a pair of leaves
attached to $u$. On the other hand, for any cut $K$ for
$R$, the vertices in $G_u$ corresponding to the leaves in $T$
that are incident to the edges in $K$ form a vertex cover for $G_u$. It follows
that the number of edges in any cut $K$ that are incident to the
leaves  corresponding to the vertices in $G_u$ is at least the size of a minimum
vertex cover for $G_u$.

\section{Reduction rules}
\label{sec:structure}

All the reduction rules, terminologies, and branching rules in this section appear in~\cite{multicut}.

Let $(T, R, k)$ be an instance of {\sc multicut on trees}. We can assume that $T$ is nontrivial (contains at least three vertices). We shall assume that $T$ is rooted at some internal vertex in the tree (chosen arbitrarily), say vertex $r$. A vertex $u \in V(T)$ is {\em important} if all the children of $u$ are leaves. For a set of vertices $V' \subseteq V(T)$ and a vertex $u \in V'$, $u$ is {\em farthest} from $r$ with respect to $V'$ if $dist_T(u, r) = \max\{dist_T(w, r) \mid w \in V'\}$.

The following reduction rules for {\sc multicut on trees} are folklore, easy to verify, and can be implemented to run in polynomial time (see~\cite{thomase,guo} for proofs). Therefore, we omit their proofs.

\begin{reduction} [{\bf Useless edge}]
\label{red:0.1}
If no request in $R$ is disconnected by the removal of edge $uv \in E(T)$, then remove edge $uv$ from $T$.
\end{reduction}


\begin{reduction} [{\bf Unit request}]
\label{red:0.3}
If $(u,v) \in R$ and $uv \in E(T)$, then cut $uv$ (i.e., remove $uv$ from $T$ and decrement $k$ by 1).
\end{reduction}

\begin{lemma}
\label{lem:minvc}
Let $(T, R, k)$ be an instance of {\sc multicut on trees}. Suppose that $T$ is rooted at $r$. There exists a minimum cut $E_{min}$ for the requests of $R$ in $T$ such that, for every important vertex $u \in V(T)$, the subset of edges in $E_{min}$ that are incident to the children of $u$ corresponds to a minimum vertex cover of $G_u$.
\end{lemma}

\begin{reduction}\label{red:3}
Let $(T, R, k)$ be an instance of {\sc multicut on trees}, where $T$ is rooted at $r$, and let $u \neq r$ be a vertex in $T$.  If there exists no request between a vertex in $V(T_u)$ and a vertex in $V(T_{\pi(u)}) \setminus V(T_u)$ then contract the edge $u\pi(u)$.
\end{reduction}

\begin{reduction}\label{red:4}
Let $(T, R, k)$ be an instance of {\sc multicut on trees}, where $T$ is rooted at $r$, and let $u$ be an important vertex in $T$ such that $\Delta(G_u) \leq 2$. If there exists a (leaf) child $l$ of $u$ that is not in any minimum vertex cover of $G_u$, then contract the edge $ul$.
\end{reduction}

\begin{reduction}\label{red:5}
Let $(T, R, k)$ be an instance of {\sc multicut on trees}, where $T$ is rooted at $r$, and let $w$ be an important vertex in $T$ such that $\Delta(G_w) \leq 2$. For every path in $G_w$ of even length, cut the leaves in $children(w)$ that correspond to the unique minimum vertex cover of $P$.
\end{reduction}

\begin{definition} \rm
Let $(T, R, k)$ be an instance of {\sc multicut on trees}, where $T$ is rooted at $r$, and let $w \neq r$ be an important vertex in $T$. A request between a vertex in $V(T_w)$ and a vertex in $V(T_{\pi(w)}) \setminus V(T_w)$ is called a {\em cross request}.
\end{definition}

\begin{reduction}\label{red:6}
Let $(T, R, k)$ be a reduced instance of {\sc multicut on trees}, where $T$ is rooted at $r$, and let $w \neq r$ be an important vertex in $T$ such that $\Delta(G_w) \leq 2$. If there is a minimum vertex cover of $G_w$ such that cutting the leaves in this minimum vertex cover cuts all the cross requests from the vertices in $V(T_w)$ then contract $w\pi(w)$.
\end{reduction}

\begin{definition}\rm
\label{def:strongreduced}
The instance $(T, R, k)$ of {\sc multicut on trees} is said to be {\em reduced} if none of the above reduction rules is applicable to the instance.
\end{definition}

\begin{proposition}(\cite{multicut})\label{prop:reduced}
Let $(T, R, k)$ be a reduced instance, where $T$ is rooted at a vertex $r$. Then the following are true:

\begin{itemize}
\item[(i)] For any vertex $u \in V(T)$, there exists no request between $u$ and $\pi(u)$.

\item[(ii)] For any vertex $u \neq r$ in $V(T)$, there exists a request between some vertex in $V(T_u)$ and some vertex in $V(T_{\pi(u)}) \setminus V(T_u)$.

\item[(iii)] For any internal vertex $u \in V(T)$, there exists at least one request between the vertices in $V(T_u) - u$.

\item[(iv)] For any important vertex $w \in V(T)$ such that $\Delta(G_w) \leq 2$ and any child $u$ of $w$, there exists a request between $u$ and a sibling of $u$, and hence all the children of an important vertex are good leaves.

\item[(v)] For any important vertex $w \in V(T)$ such that $\Delta(G_w) \leq 2$, $G_w$ contains no path of even length.

\item[(vi)] For every leaf $l \in V(T)$, there exists a minimum vertex cover of $G_{\pi(l)}$ that contains $l$.

\item[(vii)] For any important vertex $w \neq r$ in $V(T)$ such that $\Delta(G_w) \leq 2$, there is no minimum vertex cover of $G_w$ such that cutting the leaves in this minimum vertex cover cuts all the cross requests from the vertices in $V(T_w)$.
\end{itemize}
\end{proposition}

\begin{observation}\rm
\label{obs:000}
If there exists a child $u$ of an important vertex $w$ such that $u$ has a cross request to its grandparent $\pi(w)$, then cut $u$.
This can be justified as follows. Any minimum cut of $T$ either cuts $w\pi(w)$ or does not cut it. If the minimum cut cuts $w\pi(w)$, then we can assume that it cuts edge $wu$ as well because by Reduction Rule~\ref{red:4}, $u$ is in some minimum vertex cover of $G_w$. On the other hand, if the minimum cut does not cut $w\pi(w)$, then it must cut edge $wu$ since $(u, \pi(w)) \in R$. It follows that in both cases there is a minimum cut that cuts $wu$. We have $L(k) \leq L(k-1)$ in this case.
\end{observation}

\begin{observation}\rm  Let $T$ be a tree rooted at $r$, let $w \neq r$ be an important vertex in $T$, and let $u$ be a child of $w$ such that $u$ is contained in some minimum vertex cover of $G_w$. If edge $w\pi(w)$ is in some minimum cut of $T$, then the edges incident to the leaves of any minimum vertex cover of $G_w$ are contained in some minimum cut: simply replace all the edges that are incident to the children of $w$ in a minimum cut that contains $w\pi(w)$ with the edges incident to the leaves corresponding to the desired minimum vertex cover of $G_w$. Since $u$ is contained in some minimum vertex cover of $G_w$, there is a minimum cut that contains the edge $wu$. Therefore, if we choose edge $w\pi(w)$ to be in the solution, then we can choose the edge $wu$ to be in the solution as well. If when we branch we choose to cut $uw$ whenever we cut $w\pi(w)$ then we say that we {\em favor} vertex $u$. Note that if we favor a vertex $u$, then by contrapositivity, if we decide not to cut $u$ in a branch, then we can assume that $w$ will not be cut as well in the same branch. This observation will be very useful when branching.
\end{observation}

\begin{observation}\rm
Let $T$ be a tree and let $w \in V(T)$ be an important vertex. Let $v \in G_w$, and recall that $deg_{G_w}(v)$ denotes the degree of $v$ in $G_w$. By Lemma~\ref{lem:minvc}, we can assume that the set of edges in $T_w$ that are contained in the solution that we are looking for corresponds to a minimum vertex cover of $G_w$. Since any minimum vertex cover of $G_w$ either contains $v$, or excludes $v$ and contains its neighbors, we can branch by cutting $v$ in the first side of the branch, and by cutting the neighbors of $v$ in $G_w$ in the second side of the branch. Note that by part (iv) of Proposition~\ref{prop:reduced}, and the fact that there is no request between a child and its parent (unit request rule), there must be at least one request between $v$ and another child of $w$, and hence, $deg_{G_w}(v) \geq 1$. \\
\end{observation}

The above observation leads to the following branching rule:

\begin{branchrule}\label{branch:1}
Let $T$ be a tree, and let $w \in V(T)$ be an important vertex. If there exists a vertex $v \in G_w$ such that $deg_{G_w}(v) \geq 3$, then branch by cutting $v$ in the first side of the branch, and by cutting the neighbors of $v$ in $G_w$ in the second side of the branch. Cutting $v$ reduces the parameter $k$ by 1, and cutting the neighbors of $v$ in $G_w$ reduces $k$ by at least 3. Therefore, the number of leaves in the search tree of the algorithm, $L(k)$, satisfies the recurrence relation: $L(k) \leq L(k-1) + L(k-3)$.
\end{branchrule}

\section{The algorithm}
\label{sec:algo}
Let $(T, R, k)$ be a reduced instance of {\sc multicut}. The algorithm is a branch-and-search algorithm, and
its execution can be depicted by a search tree. The running time
of the algorithm is proportional to the number of root-leaf paths,
or equivalently, to the number of leaves in the search tree,
multiplied by the time spent along each such path, which will be polynomial in $k$. Therefore, the
main step in the analysis of the algorithm is to derive an upper
bound on the number of leaves $L(k)$ in the search tree. We shall assume that the instance $(T, R, k)$ is reduced before every branch of the algorithm. We shall also assume that the branches are considered in the listed order. In particular, when a branch is considered, $(T, R, k)$ is reduced and none of the branches in the previous section applies.

We can now assume from the previous section that for any important vertex $w$, we have $\Delta(G_w) \leq 2$, and hence, $G_w$ consists of a collection of disjoint paths and cycles.
Moreover, we can assume that, for any important vertex $w$, no child of $w$ has a cross request to $\pi(w)$ (if it exists). We draw another observation:

\begin{observation}
\label{obs:oddpath}
If for an important vertex $w$, $G_w$ contains a path $P$ of odd length whose length is more than 3, let $u$ be an endpoint of $P$ (i.e., a vertex of degree 1 in $P$). Observe that there exists exactly one minimum vertex cover $C_u$ of $P$ containing $u$. Therefore, by Lemma~\ref{lem:minvc}, if we decide to cut $u$, then we can cut $|C_u| = (|P|+1)/2 \geq 3$ edges between $w$ and the vertices in $C_u$. On the other hand, if $wu$ is kept then the neighbor of $u$ in $G_w$ is cut.
\end{observation}

The above observation leads to the following branching rule:

\begin{branchrule}\label{branch:3}
Let $T$ be a tree, and let $w \in V(T)$ be an important vertex such that $\Delta(G_w) \leq 2$. If there exists a path $P$ in $G_w$ of odd length such that $|P| > 3$, let $u$ be an endpoint of $P$ and let $C_{u}$ be the (unique) minimum vertex cover of $P$ containing $u$. Branch by cutting the vertices in $C_{u}$ in the first side of the branch, and by cutting the neighbor of $u$ in $P$ in the second side of the branch. Since $|C_u| = (|P|+1)/2 \geq 3$, $L(k)$ satisfies the recurrence relation: $L(k) \leq L(k-3) + L(k-1)$.
\end{branchrule}

Now for any important vertex $w$, $G_w$ consists of a collection of disjoint cycles and paths of lengths 3 or 1 (i.e., edges). Note that every vertex in $G_w$ is contained in some minimum vertex cover of $G_w$. Let $T$ be a tree rooted at $r$, and let $w \in T$ be an important vertex that is farthest from $r$. We distinguish the following cases when branching. The cases are considered in the listed order, and we shall assume that $T$ is reduced and none of BranchRule~\ref{branch:1} and BranchRule~\ref{branch:3} is applicable before any of the cases.

\begin{case}\label{case:0}
Vertex $w$ has a cross request to a non-leaf sibling $w'$.
\end{case}

In this case at least one of $w, w'$ must be cut. We branch by cutting $w$ in the first side of the branch, and cutting $w'$ in the second side of the branch. Note that by part (iii) of Proposition~\ref{prop:reduced}, the size of a minimum vertex cover in $G_w$ is at least 1, and similarly for $G_{w'}$ because $w'$ is a non-leaf vertex. Moreover, a minimum vertex cover for each of $G_w$ and $G_{w'}$ can be computed in polynomial time since both graphs have maximum degree at most 2 (note that by the choice of $w$, $w'$ is an important vertex as well). Therefore, in the first side of the branch we end up cutting the edges corresponding to a minimum vertex cover of $G_w$, which reduces the parameter further by at least 1. Similarly, we end up reducing the parameter further by at least 1 in the second side of the branch. Therefore, we have $L(k) \leq 2L(k-2)$ in this case.

\begin{case}\label{case:2}
There exists a child $u$ of $w$ such that $deg_{G_w}(u) =2$ and $u$ has a cross request.
\end{case}

We favor $u$. Note that since we can assume that the solution contains a minimum vertex of $G_w$, we can branch by cutting $u$ in the first side of the branch, and by keeping $u$ and cutting the two neighbors of $u$ in $G_w$ in the second side of the branch.

If the cross request is between $u$ and an uncle $w'$ of $u$, then we branch as follows. In the first side of the branch we cut $u$. In the second side of the branch we keep edge $uw$, and cut the two neighbors of $u$ in $G_w$. Since $u$ is not cut and $u$ is favored, $w$ is not cut as well, and hence $w'$ must be cut. Therefore, $L(k)$ in this case satisfies the recurrence relation $L(k) \leq L(k-3) + L(k-1)$.

If the cross request is between $u$ and a cousin $u'$ of $u$, let $w'=\pi(u')$ and note that $\pi(w) = \pi(w')$. We favor $u'$; thus if $u'$ is not cut then $w'$ is not cut as well. In this case we branch as follows. In the first side of the branch we cut $u$. In the second side of the branch $uw$ is kept and we cut the two neighbors of $u$ in $G_w$. Since in the second side of the branch $uw$ is kept, $w\pi(w)$ is kept as well, and $u'$ must be cut (otherwise, $w'$ is not cut as well because $u'$ is favored) since $(u, u') \in R$. Therefore, $L(k)$ in this case satisfies the recurrence relation $L(k) \leq L(k-1) + L(k-3)$.

\begin{case}\label{case:2.5}
There exists a child $u$ of $w$ such that $u$ is an endpoint of a path of length 3 in $G_w$ and $u$ has a cross request.
\end{case}

Let the path containing $u$ in $G_w$ be $P=(u, x, y, z)$. We favor $u$. Note that since we can assume that the solution contains a minimum vertex of $G_w$, we can branch by cutting $u$ in the first side of the branch, and in this case $y$ can be cut as well, and by cutting $x$ in the second side of the branch.

If the cross request is between $u$ and an uncle $w'$ of $u$, then we branch as follows. In the first side of the branch we cut $u$ and $y$. In the second side of the branch we keep $uw$ and cut $x$. Since $u$ is not cut in the second side of the branch and $u$ is favored, $w$ is not cut as well, and hence $w'$ must be cut. Therefore, $L(k)$ in this case satisfies the recurrence relation $L(k) \leq 2L(k-2)$.

If the cross request is between $u$ and a cousin $u'$ of $u$, let $w'=\pi(u')$ and note that $\pi(w) = \pi(w')$. We favor $u'$; thus if $u'$ is not cut then $w'$ is not cut as well. In this case we branch as follows. In the first side of the branch we cut $u$ and $y$, and in the second side of the branch $uw$ is kept and we cut $x$. Since $uw$ is kept in the second side of the branch, $w\pi(w)$ is kept as well, and $u'$ must be cut (otherwise $w'$ is not cut) since $(u, u') \in R$. Therefore, $L(k)$ in this case satisfies the recurrence relation $L(k) \leq 2L(k-2)$.

\begin{case}\label{case:1.5}
There exists a child $u$ of $w$ such that $u$ has a cross request to a non-leaf uncle $w'$.
\end{case}

Let $v$ be the neighbor of $u$ in $G_w$, and note that $uv$ must be an isolated edge in $G_w$, and hence, exactly one of $u, v$ is in any minimum vertex cover of $G_w$. We favor $u$. We branch by cutting $u$ in the first side of the branch, and cutting $v$ in the second side of the branch. In the second side of the branch $wu$ is kept, and so is $w\pi(w)$. Since $(u, w') \in R$, $w'$ must be cut. By part (iii) of Proposition~\ref{prop:reduced}, the size of a minimum vertex cover of $G_{w'}$ is at least 1, and by the choice of $w$, $w'$ is a farthest vertex from the $r$, and hence $\Delta(G_{w'}) \leq 2$. Therefore, a minimum vertex cover for $G_{w'}$ has size at least 1 and can be computed in polynomial time. It follows that the parameter is reduced by at least 3 in the second side of the branch. We have $L(k) \leq L(k-1) + L(k-3)$ in this case.

Let us summarize what we have at this point. If all the previous cases do not apply, then we can assume that, for any important node $w$ that is farthest from the root $r$ of $T$, no child of $w$ is of degree 2 in $G_w$ and no endpoint of a path of length 3 in $G_w$ has any cross requests. Therefore, no child of $w$ that belongs to a cycle or a path of odd length $\geq 3$ in $G_w$ has any cross requests. The only children of $w$ that may have cross requests are the endpoints of the isolated edges in $G_w$. Moreover, if $w$ has a cross request then it must be to a leaf-sibling, and if a child of $w$ has a cross request to an uncle, then it must be to a leaf-uncle.

\begin{case}\label{case:3}
There exists a child $u$ of $w$ such that $u$ has at least 2 cross requests.
\end{case}

By the above discussion we have $deg_{G_w}(u) =1$. Let $v$ be the neighbor of $u$ in $G_w$, and note that exactly one of $u, v$ is in any minimum vertex cover of $G_w$. Let $u'$ and $u''$ be two vertices that $u$ has cross requests to. We distinguish the following subcases:

\begin{subcase}\label{subcase:31}
$\pi(u') \neq \pi(u'')$ or $\pi(u') = \pi(u'') = \pi(w)$.
\end{subcase}

We favor vertex $u$ and the vertices in $\{u', u''\}$ that are not children of $\pi(w)$, and branch as follows.

In the first side of the branch we cut $v$ and keep edge $wu$. Since edge $uw$ is kept and $u$ is favored, edge $w\pi(w)$ is kept as well. Since the vertices in $\{u', u''\}$ that are not children of $\pi(w)$ are favored, $u'$ and $u''$ are cut. In the second side of the branch we cut $u$. This gives $L(k) \leq L(k-1) + L(k-3)$.

\begin{subcase}\label{subcase:32}
$\pi(u') = \pi(u'')= w'$.
\end{subcase}

If there exists a minimum vertex cover of $G_{w'}$ containing both $u'$ and $u''$, then we favor $\{u', u''\}$ and branch as follows. In the first side of the branch we cut $v$. In this case $wu$ is kept, and so is $w\pi(w)$. Moreover, $u'$ and $u''$ are cut. In the second side of the branch $u$ is cut. This gives $L(k) \leq L(k-1) + L(k-3)$.

If there does not exist a minimum vertex cover of $G_{w'}$ containing both $u'$ and $u''$, then since $T$ is reduced and $w'$ is an important vertex, by part (v) of Proposition~\ref{prop:reduced}, $u'$ and $u''$ must be neighbors in $G_{w'}$. We favor $u$ and branch as follows. In the first side of the branch we cut $v$ and keep $wu$, and in the second side of the branch we cut $u$. When we keep $wu$ in the first side of the branch $w\pi(w)$ is kept as well. Since at least two edges in $\{\pi(w')w', w'u', w'u''\}$ must be cut (since $(u, u'), (u, u''), (u', u'') \in R$ and $uw, w\pi(w)$ are kept), it is safe to cut edges $w'\pi(w')$ and any of the two edges $w'u', w'u''$. This gives $L(k) \leq L(k-1) + L(k-3)$.

We can assume henceforth that every child of $w$ has at most 1 cross request.

\begin{case}\label{case:4}
Vertex $w$ has a cross request to a leaf sibling $w'$, and the size of a minimum vertex cover of $G_w$ is at least 2.
\end{case}

In this case at least one of the edges $w\pi(w), w'\pi(w)$ must be cut. We branch by cutting $w$ in the first side of the branch, and cutting $w'$ in the second side of the branch. Since the size of a minimum vertex cover of $G_w$ is at least 2, in the first side of the branch we can cut the edges corresponding to a minimum vertex cover of $G_w$, which reduces the parameter further by at least 2. Therefore, we have $L(k) \leq L(k-3) + L(k-1)$ in this case.

Now we can assume that if an important vertex $w$ has a cross request to a leaf sibling, then $G_w$ consists of a single edge.

\begin{case}\label{case:5}
Vertex $w$ has a cross request to a leaf sibling $w'$, and either $w$ has a request to a sibling $w'' \neq w'$ or a child $u$ of $w$ has a cross request to a vertex other than $w'$.
\end{case}

Suppose that $w$ has a cross request to a sibling $w'' \neq w'$. Then branch by cutting $w$ in the first side of the branch, and cutting both $w'$ and $w''$ in the second side of the branch. Observing that when $w$ is cut the parameter is reduced further by 1 due to cutting one of the two children of $w$ (arbitrarily chosen), we obtain $L(k) \leq L(k-3) + L(k-1)$ in this case.

If $u$ has a cross request to a sibling $w'' \neq w'$ of $w$, then we favor $u$ and branch by cutting $u$ in the first
side of the branch, and keeping $uw$ and cutting $v$ in the second side of the branch. In the second side of the branch, $w$ is kept (since $u$ is kept and is favored), and hence both $w'$ and $w''$ must be cut. We obtain $L(k) \leq L(k-3) + L(k-1)$ in this case. Similarly, if $u$ has a cross request to a cousin $x$, then we favor both $u$ and $x$. In the first side of the branch $u$ is cut, and in the second side of the branch $v, w', x$ are all cut. We obtain $L(k) \leq L(k-3) + L(k-1)$.

\begin{case}\label{case:6}
Vertex $w$ has a cross request to a leaf sibling $w'$ and $w'$ has a request to a vertex in $V(T_{\pi(w')})$ that is not a child of $w$.
\end{case}

If $w'$ has a request to a sibling $w'' \neq w'$, then we branch by cutting $w'$ in the first side of the branch and cutting both $w$ and $w''$ in the second side of the branch. Observing that when $w$ is cut the parameter is further reduced by 1 due to cutting one of the two children of $w$ (arbitrarily chosen), we obtain $L(k) \leq L(k-1) + L(k-3)$ in this case.

If $w'$ has a request to a vertex $x$ that is a nephew of $w'$, then we favor $x$. We branch by cutting $w'$ in the first side of the branch, and cutting both $w$ and $x$ in the second side of the branch. Observing that when $w$ is cut the parameter is further reduced by 1, we obtain $L(k) \leq L(k-1) + L(k-3)$ in this case.

Now if an important vertex $w$ that is farthest from $r$ has a cross request to a vertex $w'$ in $T_{\pi(w)}$, then $w'$ must be a leaf-sibling of $w$ (note that by Case~\ref{case:1.5} and by symmetry, $w$ does not have a cross request to a nephew) and: (1) $w$ has exactly two children, (2) $w$ has no request to any vertex in $V(T_{\pi(w})$ except to $w'$, (3) both children of $w$ have cross requests only to $w'$ (note that by part (vii) of Proposition~\ref{prop:reduced} both children of $w$ must have cross requests in this case), and (4) $w'$ has no request to any vertex in $V(T_{\pi(w')}) \setminus V(T_w)$. We call such a set of four vertices $\{w, w', u, v\}$ a {\em special quadruple}. The structure of a special quadruple is depicted in Figure~\ref{fig:quadruple}.

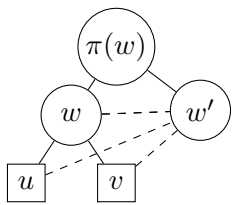
\begin{figure}
\begin{center}
\begin{tikzpicture}[>=stealth,looseness=.5,auto]
\matrix [matrix of nodes,
  column sep={0.6cm,between origins},
  row sep={0.9cm,between origins},
  nodes={circle, draw, inner sep=1.3pt}]
  {
    & & \important{p}{$\pi(w)$} \\
    & \important{w}{$w$} & & \important{w'}{$w'$} \\
    \leaf{u}{$u$} && \leaf{v}{$v$} \\
  };
  \tikzstyle{every node}=[font=\tiny\itshape]
  \draw (p) -- (w);
  \draw (p) -- (w');
  \draw[dashed] (w) -- (w');
  \draw (u) -- (w);
  \draw (v) -- (w);
  \draw [dashed] (u) -- (w');
  \draw [dashed] (w) -- (w');
  \draw [dashed] (v) -- (w');
\end{tikzpicture}
\end{center}

\caption{A special quadruple $\{w, w', u, v\}$.} \label{fig:quadruple}
\end{figure}

\begin{case}\label{case:70}
A leaf-sibling $w'$ of $w$ that is not contained in a special quadruple has at least three requests to leaf siblings or nephews.
\end{case}

If $w'$ has at least three requests to leaf siblings $w_1, w_2, w_3$, then we can branch by cutting $w'$ in the first side of the branch, and cutting all of $w_1, w_2, w_3$ in the second side of the branch. This gives $L(k) \leq L(k-1) + L(k-3)$.

If $w'$ has two requests to nephews $x$ and $y$ such that $x$ and $y$ have the same parent $w''$ and $(x, y) \in R$, let $z$ be a vertex other than $x$ and $y$ that $w'$ has a request to; if $z$ is a nephew of $w'$ then favor it. Branch by cutting $w'$ in the first side of the branch, and cutting $w''$, one of $x, y$, and $z$ in the second side of the branch (note that since Case~\ref{case:2} does not apply, $d_G(z) =1$). The reason why we can cut $w''$ in the second side of the branch follows from the fact that we would need to cut both $x$ and $y$ otherwise. This gives $L(k) \leq L(k-1) + L(k-3)$.

Finally, if the above does not apply, then we can favor all nephews of $w'$ that $w'$ has requests to, and branch by cutting $w'$ in the first side of the branch, and by cutting the siblings and nephews that $w'$ has requests to in the second side of the branch. This gives $L(k) \leq L(k-1) + L(k-3)$.

Now we can assume that for any important vertex $w$, any leaf-sibling of $w$ is either contained in a special quadruple, or has at most two requests to vertices in $V(T_{\pi(w)})$.

\begin{case}\label{case:7}
There exist two edges $uv$ and $xy$ in $G_w$ such that all vertices $u, v, x, y$ have cross requests.
\end{case}

Note that by Case~\ref{case:3}, each of $u, v, x, y$ has exactly one cross request. Moreover, by Case~\ref{case:1.5}, if there is a cross request from any of $u, v, x, y$ to an uncle, then the uncle is a leaf uncle.  Suppose that $u$ has a cross request to $u'$, $v$ to $v'$, $x$ to $x'$, and $y$ to $y'$. We distinguish the following subcases.

\begin{subcase}
\label{subcase:71}
Both $u$ and $v$ (or $x$ and $y$) have requests to the same uncle (i.e., $u'=v'$ is a sibling of $w$).
\end{subcase}

Note that in this case $u'$ is a leaf uncle. Branch by cutting $u'$ in the first side of the branch, and keeping $u'$ and cutting $w$ and the edges between $w$ and the vertices of any minimum vertex cover of $G_w$ in the second side of the branch (otherwise, if $w$ is kept then both $u$ and $v$ would need to be cut). This gives $L(k) \leq L(k-1) + L(k-3)$.

\begin{subcase}
\label{subcase:72}
$u$ and a vertex in $x, y$, say $x$, have requests to the same uncle, and $v$ and $y$ have requests to the same uncle (i.e., $u'=x'$ and $v'=y'$, where both $u'$ and $v'$ are leaf uncles of $u$).
\end{subcase}
In this case the only requests involving $u, v, x, y, w, u', v'$ are $\{(u, v), (x, y), (u', u), (u', x), (v', v), (v', y)\}$, which can be cut by cutting three edges (e.g., cutting $w, u, x$). Note that if none of $u', w, v'$ is cut, then four edges are needed to cut the above requests; similarly, if two vertices in $u', w, v'$ are cut then four edges are needed to cut the above requests. Therefore, we can conclude that there exists a minimum cut that cuts exactly one vertex in $u', w, v'$ (if a minimum cut cuts two or more vertices from $u', w, v'$ then such a cut can be replaced by another cut of the same cardinality that cuts $w$, $\pi(w)$, $u$, and $v$; similarly if a cut does not cut any of $u', v', w$).  Branch by cutting $w$ and the vertices of any minimum vertex cover of $G_w$ in the first side of the branch, cutting $u', v, y$ and keeping $v'$ and $w$ in the second side of the branch, and cutting $v', u, x$ and in the third side of the branch. This gives $L(k) \leq 3L(k-3)$.

\begin{subcase}
\label{subcase:73}
Vertex $u$ and a vertex in $x, y$, say $x$, have requests to the same uncle.
\end{subcase}

Since $\{y, v\}$ is contained in some minimum vertex cover of $G_w$, we can favor $\{y, v\}$. Since any minimum vertex cover of $G_w$ must contain either $\{u, x\}$, $\{u, y\}$,
$\{v, x\}$, or $\{v, y\}$, by Lemma~\ref{lem:minvc}, there exists a minimum cut that either cuts $u$ and $x$, or $u$ and $y$, or $v$ and $x$, or $v$ and $y$. Therefore, we can branch in a 4-way branch by cutting $u$ and $x$ in the first side of the branch, $u$ and $y$ in the second side of the branch, $v$ and $x$ in the third side of the branch, and $v$ and $y$ in the fourth side of the branch. Since $y, v$ are favored, in the first, second, and third sides of the branch $w\pi(w)$ is kept. Observe the following. First, if a vertex in $u, v, x, y$ has a request to a leaf uncle, then in any of the first three sides of the branch in which the vertex is not cut the uncle must be cut. Observe also that in the first side of the branch, if the two vertices $v, y$ that are not cut have cross requests to two cousins $v', y'$, respectively, such that $(v', y') \in R$, then $\pi(v')$ can be cut (otherwise, both $v'$ and $y'$ need to be cut) in addition to one of $v', y'$; if this is not the case then we can favor $\{y', v'\}$. Based on the above, the parameter is reduced by 4 in the first side of the branch, 4 in the second side of the branch, 4 in the third side of the branch, and 2 in the fourth side of the branch. We get
$L(k) \leq 3L(k-4) + L(k-2)$.

\begin{subcase}
\label{subcase:74}
No two requests from $u, v, x, y$ go to the same uncle.
\end{subcase}

Similarly to Subcase~\ref{subcase:73}, we can favor two vertices in $u, v, x, y$, chosen arbitrarily, say $v, y$, and branch in a 4-way branch by cutting $u$ and $x$ in the first side of the branch, $u$ and $y$ in the second side of the branch, $v$ and $x$ in the third side of the branch, and $v$ and $y$ in the fourth side of the branch. Since $y, v$ are favored, in the first, second, and third sides of the branch $w\pi(w)$ is kept. By drawing the same observations as in Subcase~\ref{subcase:73}, we conclude that the first three sides of the branch result in a reduction of the parameter by a value of at least 4, whereas the fourth side of the branch results in a reduction of the parameter by a value of at least 2. We get $L(k) \leq 3L(k-4) + L(k-2)$.

The following proposition follows from the inapplicability of the above cases plus the fact that $T$ is reduced:

\begin{proposition}
\label{prop:properties}
Let $T$ be a reduced tree with root $r$, and let $w \neq r \in T$ be an important vertex that is farthest from $r$. If none of BranchRule~\ref{branch:1}, BranchRule~\ref{branch:3} and the above cases applies, then the following hold true:

\begin{itemize}

\item[(i)] For every child $w'$ of $\pi(w)$ (i.e., sibling of $w$ or $w$ itself) that is an important vertex, $G_{w'}$ consists of disjoint edges, length-3 paths, and cycles. No vertex that is contained in a cycle or a length-3 path in $G_{w'}$ has any cross requests, and every endpoint of an edge in $G_{w'}$ has at most one cross request.

\item[(ii)] For every child $w'$ of $\pi(w)$ that is an important vertex, there exist exactly two children $u, v$ of $w'$ such that $(u, v) \in R$ and both $u$ and $v$ have cross requests.

\item[(iii)] Every leaf child $w'$  of $\pi(w)$ that is not contained in a special quadruple has at least one request, and at most two requests, to vertices in $V(T_{\pi(w)})$ that are either leaf siblings or nephews of $w'$.

\item[(iv)] Every non leaf child of $\pi(w)$ that is not contained in a special quadruple has no cross requests.
\end{itemize}
\end{proposition}

\begin{proof}

\begin{itemize}

\item[(i)] We know that $\Delta(G_{w'}) \leq 2$. Therefore, $G_{w'}$ consists of disjoint paths and cycles. By part (v) of Proposition~\ref{prop:reduced}, $G_{w'}$ contains no path of even length, and by BranchRule~\ref{branch:3}, $G_{w'}$ contains no path of odd length $\geq 5$. Therefore, $G_{w'}$ consists of disjoint edges, length-3 paths, and cycles. By Case~\ref{case:2}, no vertex in $G_{w'}$ of degree 2 has a cross request, and hence the only vertices in $G_{w'}$ that can have cross requests are endpoints of paths (or disjoint edges). By Case~\ref{case:2.5}, no endpoint of a length-3 path in $G_{w'}$ has a cross request, and hence no vertex of a length-3 path has a cross request. By Case~\ref{case:3}, no vertex in $G_{w'}$ has two cross requests, and hence every endpoint of a disjoint edge has at most one cross request. The statement follows.

\item[(ii)] By part (i) above, the only vertices in $G_{w'}$ that can have cross requests are endpoints of disjoint edges. By part (vii) of Proposition~\ref{prop:reduced}, there is no minimum vertex cover of $G_{w'}$ that cuts all cross requests. Therefore, there must exist at least one disjoint edge in $G_{w'}$ whose both endpoints have cross requests. By Case~\ref{case:7}, such an edge must be unique. The statement follows.

\item[(iii)] The fact that a leaf child of $w'$ must have at least one cross request follows from part (ii) (and part (i)) of Proposition~\ref{prop:reduced}. By Case~\ref{case:70}, no leaf child of $w'$ that is not contained in a special quadruple can have more than 2 cross requests. The statement follows.

\item[(iv)] Let $w'$ be a non leaf child of $\pi(w)$ that has a cross request to a vertex $w''$, and we show that $w'$ must be contained in a special quadruple. By Case~\ref{case:0}, the cross request from $w'$ must be to a leaf sibling. By Case~\ref{case:4}, the size of a minimum vertex cover of $G_{w'}$ is exactly 1 (note that every vertex in $G_{w'}$ has degree at least 1), and hence $G_{w'}$ consists of a single edge $uv$. By part (ii) of the current proposition, both $u$ and $v$ have cross requests, and by Case~\ref{case:1.5}, the requests from $u$ and $v$ must be to leaf uncles. Also, by Case~\ref{case:3}, each of $u$ and $v$ has exactly one cross request. By Case~\ref{case:5}, $w'$ must have exactly one cross request to $w''$ and the cross requests from both of $u$ and $v$ must be to $w''$. Finally, by Case~\ref{case:6}, the cross requests from $w''$ must be only to $w$, $u$, and $v$. The statement follows.

\end{itemize}
\end{proof}

\begin{definition}\rm
\label{def:auxiliary1}
Let $T$ be a reduced tree with root $r$, and let $w \neq r \in T$ be an important vertex that is farthest from $r$. Suppose that none of BranchRule~\ref{branch:1}, BranchRule~\ref{branch:3}, or the above Cases applies. We define the auxiliary graph $G^*_{\pi(w)}$ as follows. The vertices of $G^*_{\pi(w)}$
are the leaf children and the grandchildren of $\pi(w)$ that are not contained in any special quadruple.  Two vertices $x$ and $y$ in $G^*_{\pi(w)}$ are adjacent if and only if $(x, y) \in R$. Note that the edges in $G^*_{\pi(w)}$ correspond to either a request between two grandchildren of $\pi(w)$ that have the same parent, a request between two leaf children of $\pi(w)$, or a request between a leaf-child and a grandchild of $\pi(w)$.
\end{definition}

The following proposition is the dual of Proposition~\ref{prop:reduced}:

\begin{proposition}
\label{prop:propertiesbranching}
Let $T$ be a reduced tree with root $r$, and let $w \in T$, where $\pi(w) \neq r$, be an important vertex that is farthest from $r$. Suppose that none of BranchRule~\ref{branch:1}, BranchRule~\ref{branch:3}, or the above Cases applies. Consider the graph $G^*_{\pi(w)}$. Then the following are true:

\begin{itemize}
\item [(a)] $\Delta(G^*_{\pi(w)}) \leq 2$, and hence $G^*_{\pi(w)}$ consists of disjoint paths and cycles.

\item[(b)] For every path $P$ in $G^*_{\pi(w)}$ such that at least one endpoint of $P$ is a grandchild of $\pi(w)$, there exists a minimum cut that cuts the vertices in some minimum vertex cover of $P$.

\item[(c)] For every path $P$ and every cycle $C$ in $G^*_{\pi(w)}$, there exists a minimum cut $C_{min}$ such that the number of edges in $C_{min}$ that are incident on the vertices in $P$ or their parents in case these vertices are grandchildren of $\pi(w)$, is equal to the size of a minimum vertex cover of $P$, and the number of edges in $C_{min}$ that are incident on the vertices in $C$ or their parents in case these vertices are grandchildren of $\pi(w)$, is equal to the size of a minimum vertex cover of $C$.

\end{itemize}
\end{proposition}

\begin{proof}

Part (a) follows from parts (i), (ii), (iii) of Proposition~\ref{prop:propertiesbranching}.

Parts (b) and (c) are similar to Lemma~\ref{lem:minvc} in spirit. Proving them, however, is more subtle since a minimum cut can cut an important vertex, which would cut all cross requests from its children, and important vertices are not vertices of $G^*_{\pi(w)}$.

Consider a minimum cut $C_{min}$ of $T$. Call a path in $G^*_{\pi(w)}$ whose both endpoints are leaf children of $\pi(w)$ a {\em type-I path}, and a path with at least one endpoint that is a grandchild of $\pi(w)$ a {\em type-II path}.

To prove part (b), consider a type-II path $P$ in $G^*_{\pi(w)}$. It is not difficult to see that any cut to $T$ must cut at least $\tau(P)$ (the size of a minimum vertex cover of $P$) many vertices of $P$. Moreover, for every vertex of $P$, all its requests to vertices in $V(T_{\pi(w)})$ are to vertices on $P$. Therefore, if $C_{min}$ cuts more than $\tau(P)$ many vertices from $P$, then the vertices of $P$ that are cut by $C_{min}$ can be replaced by those in a minimum vertex cover of $P$ plus vertex $\pi(w)$; this will result in a minimum cut of $T$ that cuts exactly $\tau(P)$ many vertices from $P$. Therefore, we can assume that $C_{min}$ cuts precisely $\tau(P)$ vertices from $P$, and it suffices to show that the vertices of $P$ that are cut by $C_{min}$ can be replaced by a vertex cover of $P$.

Suppose that $P=(u_1, u_2, \ldots, u_i)$, where $u_1$ is a grandchild of $\pi(w)$, and let $S$ be the set of vertices in $P$ that are grandchildren of $\pi(w)$; note that $u_1 \in S$. If $C_{min}$ does not cut any parent of a vertex in $S$, then it can be readily seen that the vertices of $P$ that are cut by $C_{min}$ must form a vertex cover of $P$ (those vertices would consist only of grandchildren and leaf children of $\pi(w)$, and hence of vertices of $G^*_{\pi(w)}$). Suppose now that $C_{min}$ cuts al least one parent of a vertex in $S$. Suppose first that $C_{min}$ cuts the parent of a vertex $v \in S$ such that $v$ is not an endpoint or a sibling of an endpoint of $P$; let $w'=\pi(v)$ (note that in this case $w'\neq \pi(u_1)$). Let $x$ be the child of $w'$ such that $(v, x) \in R$. By part (ii) of Proposition~\ref{prop:properties}, the edge $w' \pi(w)$ does not cut any request on a type-I path or a cycle in $G^*_{\pi(w)}$ because $v, x$ are the only children of $w'$ that both have cross requests and such that $(x, v) \in R$. Therefore, the edge $w'\pi(w)$ only cuts type-II paths of the form $(y_1, y_2, \ldots, y_j)$ where both $y_1$ and $y_2$ are children of $w'$. For each such path $(y_1, y_2, \ldots, y_j)$, if $C_{min}$ contains $y_1$ then swap it with $y_2$; we still get a minimum cut, say $C_{min}$ without loss of generality, such that $C_{min} - w'\pi(w)$ cuts all requests on type-II paths of the form  $(y_1, y_2, \ldots, y_j)$ where $y_1$ and $y_2$ are children of $w'$. Repeating the above for each such vertex $v$, and replacing the edges in $C_{min}$ that are incident on vertices of $P$ with edges that are incident on some minimum vertex cover of $P$, and replacing edge $w'\pi(w)$ in $C_{min}$ with edge $\pi(w)\pi(\pi(w))$, yields a minimum cut that cuts the vertices in a minimum vertex cover of $P$.

We can assume now that $C_{min}$ cuts the parent of an endpoint of $P$, say $w'=\pi(u_1)$. Then $C_{min}$ does not cut any parent of a vertex $v \notin \{u_1, u_2\}$ in $S$, and $C_{min}$ must cut a vertex cover of the subpath $(u_3, \ldots, u_i)$. Moreover, $C_{min}$ must cut either $u_1$ or $u_2$. If $C_{min}$ cuts $u_1$ then replace $u_1$ with $u_2$ to obtain a minimum cut that cuts a minimum vertex cover of $P$; otherwise, $C_{min}$ cuts a minimum vertex cover of $P$. The case is similar if $C_{min}$ cuts an endpoint of $P$ that is a leaf child of $\pi(w)$. This proves part (b).

To prove part (c), consider a type-I path or a cycle in $G^*_{\pi(w)}$. Observe that by part (a) above and by part (ii) of Proposition~\ref{prop:properties}, no edge in $C_{min}$ that cuts a request on a type-I path or a cycle in $G^*_{\pi(w)}$ cuts any other request on a type-I path or cycle. (Note that its is possible that an edge that cuts a request on a type-I path or cycle cuts a request on a type-II path. However, this will not affect the fact that $C_{min}$ cuts exactly the size of a minimum vertex cover many vertices from every type-II path). Therefore, if $C_{min}$ cuts more than $\tau(P)$ many vertices from a type-I path $P$ in $G^*_{\pi(w)}$, or more than $\tau(C)$ from a cycle $C$ in $G^*_{\pi(w)}$, then those vertices that are cut by $C_{min}$ can be replaced by the vertices in a minimum vertex cover of $P$ or $C$, plus vertex $\pi(w)$, to yield a minimum cut of $T$ that cuts exactly $\tau(P)$ many vertices from every path $P$ and $\tau(C)$ many vertices from every cycle $C$ in $G^*_{\pi(w)}$. This completes the proof.
\end{proof}

The following reduction rule follows from parts (b) and (c) of Proposition~\ref{prop:propertiesbranching} after noticing that for a path of even length, there is a unique set of edges of cardinality $\tau(P)$ in $E(T_{\pi(w)})$ that cuts all requests corresponding to the edges of $P$:

\begin{reduction}\label{red:10}
Let $T$ be a reduced tree with root $r$, and let $w \in T$ be an important vertex that is farthest from $r$, and such that $\pi(w) \neq r$. Suppose that none of BranchRule~\ref{branch:1}, BranchRule~\ref{branch:3}, or the above Cases applies. If there exists a path $P$ in $G^*_{\pi(w)}$ of even length then cut the vertices in $P$ that correspond to the unique vertex cover of $P$.
\end{reduction}

The following branching rule follows from parts (b) and (c) of Proposition~\ref{prop:propertiesbranching}, after noticing that for a path of odd length in $G^*_{\pi(w)}$, there is a unique set of edges of cardinality $\tau(P)$ in $E(T_{\pi(w)})$ that cuts all requests corresponding to the edges of $P$ in addition to cutting an endpoint of $P$:

\begin{branchrule}\label{branch:4}
Let $T$ be a reduced tree with root $r$, and let $w \in T$ be an important vertex that is farthest from $r$, and such that $\pi(w) \neq r$. Suppose that none of BranchRule~\ref{branch:1}, BranchRule~\ref{branch:3}, or the above Cases applies. If there exists a path $P$ in $G^*_{\pi(w)}$ of odd length such that $|P| > 3$, let $u$ be an endpoint of $P$ and let $C_{u}$ be the (unique) minimum vertex cover of $P$ containing $u$. Branch by cutting the vertices in $C_{u}$ and contracting the edges between $w$ and vertices in $V(P) - C_u$ in the first side of the branch, and by cutting the neighbor of $u$ in $P$ in the second side of the branch. Since $|C_u| = (|P|+1)/2 \geq 3$, $L(k)$ satisfies the recurrence relation: $L(k) \leq L(k-3) + L(k-1)$.
\end{branchrule}

The following branching rule follows from parts (b) and (c) of Proposition~\ref{prop:propertiesbranching} after noticing that, for a cycle of even length in $G^*_{\pi(w)}$ there are exactly two sets of edges in $E(T_{\pi(w)})$, each of cardinality $\tau(P)$, such that each cuts all requests corresponding to the edges of $C$:

\begin{branchrule}\label{branch:5}
Let $T$ be a reduced tree with root $r$, and let $w \in T$ be an important vertex that is farthest from $r$, and such that $\pi(w) \neq r$. Suppose that none of BranchRule~\ref{branch:1}, BranchRule~\ref{branch:3}, or the above Cases applies. If there exists a cycle $C$ in $G^*_{\pi(w)}$ of even length, branch into a two sided branch: in the first side of the branch cut the vertices corresponding to one of the minimum vertex covers of $C$, and in the second  side of the branch cut the vertices corresponding to the other minimum vertex cover of $C$.  Since $|C| \geq 4$, and hence $\tau(C) \geq 2$, we get $L(k) \leq 2L(k-2)$.
\end{branchrule}

\begin{branchrule}\label{branch:5}
Let $T$ be a reduced tree with root $r$, and let $w \in T$ be an important vertex that is farthest from $r$, and such that $\pi(w) \neq r$. Suppose that none of BranchRule~\ref{branch:1}, BranchRule~\ref{branch:3}, or the above Cases applies. If there exists a cycle of odd length $C=(u_1, u_2, \ldots, u_{2\ell+1})$ in $G^*_{\pi(w)}$ such that $\ell \geq 3$ and $C$ is not a cycle in $G_{w'}$ for some child $w'$ of $\pi(w)$, then branch as follows. First observe that since $C$ is not a cycle in $G_{w'}$ for some child $w'$ of $\pi(w)$, $|C|$ is odd, and $T$ is reduced, at least one vertex, say $u_1$ on $C$ must be a leaf child of $\pi(w)$. We favor the vertices in $\{u_2, u_{2\ell+1}\}$ that are grandchildren of $\pi(w)$ (if any). In the first branch $u_1$ is kept and $u_2, u_{2\ell+1}$ are cut. In the second side of the branch $u_1$ is cut, and the cycle becomes a path of odd length at least 5; therefore, we can further branch according to BranchRule~\ref{branch:4}. This yields $L(k) \leq 2L(k-2) + L(k-4)$.
\end{branchrule}

Now let $T$ be a reduced tree with root $r$, and let $w \in T$ be an important vertex that is farthest from $r$, and such that $\pi(w) \neq r$. Suppose that none of the branching rules of the above cases applies. Then each vertex in $V(T_{\pi(w)}) - \pi(w)$ is contained in one of the following structures/groups in $G^*_{\pi(w)}$. Group I, abbreviated $GP_1$, are paths of length 1 (edge) in $G^*_{\pi(w)}$ between two children of an important child of $\pi(w)$, Group II, abbreviated $GP_2$, are paths of length 1 in $G^*_{\pi(w)}$ between two leaf children of $\pi(w)$, Group III, abbreviated $GP_3$, are paths of length 3 in $G^*_{\pi(w)}$ but not in $G_w'$ for any important child of $\pi(w)$, Group IV, abbreviated $GP_4$ are cycles of length 3 in $G^*_{\pi(w)}$ but not in $G_w'$ for any important child of $\pi(w)$, Group V, abbreviated $GP_5$, are cycles of lengths 5 in $G^*_{\pi(w)}$  but not in $G_w'$ for any important child of $\pi(w)$, Group VI, abbreviated $GP_6$ are special quadruples, and Group VII, abbreviated $GP_7$, are paths of length $3$ or cycles in $G_{w'}$, for some important child $w'$ of $\pi(w)$. Note that no vertex in $GP_7$ can have a cross request. The structure of the groups are illustrated in Figures~\ref{fig:1}~--~\ref{fig:6}, in addition to Figure~\ref{fig:quadruple}.

\include{figures}

Let $w_1 = \pi(w)$, and let $w_2, \ldots, w_l$ be the siblings of $w_1$. Each sibling $w_i$ of $w_1$ is either a leaf, an important vertex, or $T_{w_i}$ has a similar structure to $T_{w_1}$.

\begin{case}
\label{case:100}
If for any $w_i$, there is no request from a vertex in $T_{w_i}$ to a vertex in some tree $T_{w_j}$, for any $j\neq i$, then contract the edge between $w_i$ and its parent.
\end{case}
This can be seen as follows. If $w_i$ is cut by some minimum cut $C_{min}$, then since there are no requests between vertices in $T_{w_i}$ and a vertex in $T_{w_j}$, for any $j \neq i$ in $\{1, \ldots, l\}$, edge $w_i\pi(w_i)$ can be replaced by the edge between $\pi(w_i)$ and its parent to yield a minimum cut that excludes the edge between $w_i$ and its parent. Therefore, we can assume that, for any $i \in \{1, \ldots, l\}$, there exists a request between some vertex in $T_{w_i}$ and a vertex in some $T_{w_j}$. Moreover, for any vertex $u$ in $T_{w_i}$, there exists a minimum cut that cuts $u$. Further, if any edge $e$ on the path between $\pi(w_i)$ and $u$ is part of a minimum cut, then there is a minimum cut that includes $e$ and cuts $u$ as well. Therefore, $u$ can be favored in the sense that if $u$ is kept in a certain branch then all edges on the path between $u$ and $w_i$ are kept as well.

Consider now $w_i$ for some fixed $i$. Let $u$ be a vertex in $T_{w_i}$ that has a request to a vertex $x$ in $T_{w_j}$, for some $j\neq i$. We favor both $u$ and $x$.  We distinguish the following cases.

\begin{case}
\label{case:200}
$u$ is an important vertex. We can branch with $L(k) \leq L(k-3) + L(k-1)$.
\end{case}

Since $u$ is important, $u$ must have two children $y, z$ such that $(y, z) \in R$ and both $z$ and $y$ have cross requests. Suppose that $z$ has a cross request to $z'$ in $T_{w_i}$. Favor $z'$ (if $z'$ is a grandchild of $w_i$) and $z$. Branch by cutting $y$ in the first size of the branch and keeping $z$, and by cutting $z$ in the second side of the branch. In the first side of the branch $z$ is kept and so is $u$. Therefore, $z'$ and $x$ must be cut. This gives $L(k) \leq L(k-3) + L(k-1)$.

We can assume now that $u$ is not an important vertex in $T_{w_i}$. Therefore, $u$ must be a vertex in $G^*_{w_i}$; let $d_u$ be the degree of $u$ in $G^*_{w_i}$.

\begin{case}
\label{case:300}
$d_u =2$.  We can branch with $L(k) \leq L(k-3) + L(k-1)$.
\end{case}

If in a certain branch $u$ is kept, then two edges can be cut. This can be seen as follows. If $u$ is a $GP_3$ vertex, then both neighbors of $u$ in $G^*_{w_i}$ can be cut (favor the neighbors that are not leaf children of $w_i$). If $u$ is a $GP_4$ vertex, let $(u, u_1, u_2)$ be the length-3 cycle containing $u$.  If $u$ is a leaf child of $w_i$, then the parent of $u_1, u_2$ can be cut, in addition to one of $u_1, u_2$ (chosen arbitrarily). If $u$ is not an leaf child of $w_i$, then $u_1$ and $u_2$ can be cut (since $u$ is favored). If $u$ is a $GP_5$ vertex, then by favoring any neighbor of $u$ in $G^*_{w_i}$ that is not a leaf child of $w_i$ (if the neighbor is a leaf child of $w_i$ then there is no need to favor it), it can be easily seen that when $u$ is kept then its two neighbors can be cut. If $u$ is a $GP_6$ vertex, then the same analysis carries as when $u$ is $GP_4$ vertex. Finally, if $u$ is a $GP_7$ vertex, then it can be easily seen that both neighbors of $u$ can be cut.

Therefore, if $d_u=2$, then we can branch by cutting $u$ in the first side of the branch, and keeping $u$ and cutting its two neighbors in $G^*_{w_i}$, in addition to $x$ in the second side of the branch. This gives $L(k) \leq L(k-1) + L(k-3)$.

\begin{case}
\label{case:400}
Suppose now that $d_u=1$ in $G^*_{w_i}$. We can branch with $L(k) \leq 2L(k-2)$.
\end{case}

In this case either $u$ is a $GP_3$ or a $GP_7$ vertex that is an endpoint of a length-3 path, or $u$ is a $GP_1$ or a $GP_2$ vertex . If $u$ is an endpoint of a length-3 path $(u, u_1, u_2, u_3)$, then by part (b) of Proposition~\ref{prop:propertiesbranching}, if $u$ is cut, then $u_2$ must be cut as well. On the other hand, if $u$ is kept then $u_1$ and $x$ must be cut. This gives $L(k) \leq 2L(k-2)$.

We can now assume that all requests between the $T_{w_i}$'s go from $GP_1$ or $GP_2$ vertices to $GP_1$ or $GP_2$ vertices.

\begin{case}
\label{case:500}
$u$ is an endpoint of a $GP_1$ group. We can branch with $L(k) \leq L(k-1) + 2L(k-4)$.
\end{case}

If $u$ is an endpoint of a $GP_1$ group, let $w=\pi(u)$, and let $v$ be the child of $w$ such that $(u, v) \in R$.  Note that $w$ is an important vertex, and hence, there exists $z, y$, children of $w$, such that $(z, y) \in R$ and both $z$ and $y$ have cross requests. We favor $u$ and branch as follows. In the first side of the branch we cut $u$ and keep $v$, and in the second side of the branch we keep $u$ and cut $v$. Let us analyze the second side of the branch when $u$ is kept. In this case $x$ must be cut. Since $u$ is kept and is favored, $w$ is kept as well. Since both $z$ and $y$ have cross requests, $z$ and $y$ are either part of a $GP_3$, $GP_4$, $GP_5$, or $GP_6$ group. If $z$ and $y$ are contained in a $GP_6$ or a $GP_4$ group, then their uncle must be cut leading to a further reduction of the parameter by at least 1. If $z$ and $y$ are contained in a $GP_3$ group $(s, z, y, t)$, then by part (b) of Proposition~\ref{prop:propertiesbranching} either $s, y$ or $z, t$ must be cut, so we can branch further into these two branches. If $z$ and $y$ are part of a $GP_5$ group $(s, z, y, p, q)$, where $p$ and $q$ are children of an important child of $w_i$, then since $w$ is kept we branch on $z$: if $z$ is cut, then $y$ is kept and $p$ is cut (since $w$ is kept), and if $z$ is kept then $y$ and $s$ are cut. In the worst case, we get $L(k) \leq L(k-1) + 2L(k-4)$.

\begin{case}
\label{case:600}
All requests between the $T_{w_i}$'s go between $GP_2$ vertices. We can branch with $L(k) \leq 2L(k-2) + L(k-4)$.
\end{case}

If for every $GP_2$ group in $T_{w_i}$ at most one vertex has a request to some $T_{w_j}$, where $j \neq i$, then there exists a cut of $T_{w_i}$ that cuts all requests to $T_{w_j}$, and whose cardinality is equal to the set of edges in a minimum cut that are contained in $T_{w_i}$; therefore, edge $w_i \pi(w_i)$ can be contracted.  Hence, we can assume that for every $T_{w_i}$, there exists a $GP_2$ in $T_{w_i}$ whose both vertices $u, v$ have requests to vertices in other trees; suppose that $u$ has a request to $u'$ and $v$ to $v'$, where $u'$ and $v'$ are not in $T_{w_i}$. Moreover, we can assume that $T_{w_i}$ contains an important vertex (choose a tree among the $T_{w_j}$'s that contains an important vertex, and by the above argument, there exists a $GP_2$ in $T_{w_j}$ whose both vertices $u, v$ have requests to vertices in other trees). Since each important vertex must have two children with cross requests, and since there is a $GP_2$ in $T_{w_i}$, any minimum cut must cut at least three vertices in $V(T_{w_i}) - w_i$. We branch as follows. Either $w_i$ is cut or is kept. When $w_i$ is cut, at least 3 edges in $E(T_{w_i})$, corresponding to any minimum cut of $T_{w_i}$ can be cut. When $w_i$ is kept, we branch by cutting $u$ and favoring $v'$ in the first side of the branch, and cutting $v$ and favoring $u'$ in the second side of the branch. When $u$ is cut, $v$ is kept, and hence $v'$ must be cut (since $w_i$ is kept). When $v$ is cut, $u$ is kept and hence $u'$ is cut. This gives $L(k) \leq 2L(k-2) + L(k-4)$.

\begin{theorem}
\label{thm:main}
The MCT problem is solvable in time $O^{*}(\rho^k)$, where $\rho = \sqrt{\sqrt{2} + 1} \approx 1.555$.
\end{theorem}

\begin{proof}
The above cases exhaust all possible scenarios. The worst branch is $L(k) \leq 2L(k-2) + L(k-4)$, which corresponds to the characteristic polynomial $x^4 - 2 x^2 -1$ whose positive root is  $\sqrt{\sqrt{2} + 1}$. It follows that the running time of the algorithm is $O^{*}(\rho^k)$, where $\rho = \sqrt{\sqrt{2} + 1} \approx 1.555$.
\end{proof}

\section{The GMWCT problem}
\label{sec:gmwct}
There is a simple reduction from GMWCT to MCT.  For an instance $(T, \{S_1, \ldots, S_r\}, k)$ of GMWCT, construct the instance $(T, R, k)$, where the set of requests $R$ is given as follows. For each terminal set $S_i$, $i=1, \ldots, r$, and for every pair of distinct terminals $u, v \in S_i$, add the request $(u, v)$ to $R$. Clearly, $(T, \{S_1, \ldots, S_r\}, k)$ is a yes-instance of GMWCT if and only if $(T, R, k)$ is a yes-instance of MCT.
Combining this reduction with Theorem~\ref{thm:main} we obtain:

\begin{theorem}
\label{thm:gmwct}
The GMWCT problem is solvable in time $O^{*}(\rho^k)$, where $\rho = \sqrt{\sqrt{2} + 1} \approx 1.555$.
\end{theorem}

\subsection{A linear time algorithm for WGMWCT}
\label{subsec:wgmwc}
Next, we show that the GMWCT problem is solvable in linear time when the number of terminal sets is a constant. (Clearly, the problem is NP-complete when the number of terminal sets is part of the input by a simple reduction from the MCT problem.) The algorithm is a dynamic programming algorithm that solves the more general weighted version of the problem, denoted WGMWCT, defined as follows: Given a tree $T$ in which each edge is associated with a nonnegative cost $c(e)$, and terminal sets $\{S_1, S_2, \dots, S_c\}$, where $c$ is a constant, compute a minimum-cost set of edges
whose removal cuts each pair of distinct terminals in $S_i$, for $i=1, \ldots, c$.

{\bf Assumptions.}
First of all, root the input tree $T$ at {\em any} vertex.
Let $r$ be the root of $T$. We first apply some simple preprocessing:
\begin{enumerate}
\item If an internal vertex $u$ is a terminal, we can add a child of $u$, say $w$, and assign the edge $(u, w)$ a large enough cost; then replace each appearance of $u$ in collection $R$ by $w$.
All leaves that are {\em not} terminals can be removed from the tree. After this preprocessing, we may assume that {\em a vertex $u$ is a leaf if and
only if $u$ is a terminal}.

\item If a vertex $u$ has more than $t>2$ children, then we can replace $u$ by a chain of $t-1$ vertices connected by a path in which each edge has a large enough cost. The original children of $u$ is then connected to the vertices in the chain so that each vertex in the chain has two children. After this preprocessing, we may also assume that {\em the input tree $T$ is a binary
tree}.
\end{enumerate}

{\bf Basic idea.}
Our approach to solve WGMWCT is dynamic programming, based on the following
simple observation. When an optimal multiway cut is removed from the input
tree $T$, $T$ is broken into several subtrees. In each subtree, there is
at most one terminal for {\em every} terminal set $S_i$ ($1 \leq i \leq q$), and
there must be at least one terminal which comes from {\em some} terminal set.

{\bf The algorithm.}
Let $u$ be a vertex in the tree $T$ and $T_u$ the subtree of $T$ rooted at $u$. When certain edges are removed in $T_u$, some terminals in the original $T_u$
are still connected to $u$ and some terminals are not. We use a binary vector $\overrightarrow{{\cal B}} = b_1b_2\ldots b_q$ to record the {\em connection pattern} of the tree, where $b_i = 1$ ($1 \leq i \leq q$) if there exists exactly one terminal from
$S_i$ that is connected to $u$ and $b_i = 0$ if there is no terminal from $S_i$ that is connected to $u$ (if two or more terminals from $S_i$ are connected to $u$, then the cut is invalid). Denote by $DP[u, \overrightarrow{{\cal B}}]$ the minimum cost of edges whose removal makes $T_u$ satisfy $\overrightarrow{{\cal B}}$. If such a set of edges does not exists, $DP[u, \overrightarrow{{\cal B}}]$ is defined to be $+\infty$. The goal of the algorithm is to compute $DP[u, \overrightarrow{{\cal B}}]$ for all vertices $u$ and all binary vectors $\overrightarrow{{\cal B}}$.

We consider three cases:

{\bf Case 1. $u$ has two children $v$ and $w$.} We will compute $DP[u, \overrightarrow{{\cal B}}]$ based on the assumption that $DP[v, \overrightarrow{{\cal B}}]$ and $DP[w, \overrightarrow{{\cal B}}]$
have already been computed. Let $\overrightarrow{{\cal B}_{v}}, \overrightarrow{{\cal B}_{w}}$ be connection patterns of $v$ and $w$ respectively, and let ${\cal C}$ be a subset of $\{uv,uw\}$. If $u$ has connection pattern $\overrightarrow{{\cal B}}$ after the edges in ${\cal C}$ are removed, we say that $\overrightarrow{{\cal B}}$ is {\em derived from} the tuple $\{\overrightarrow{{\cal B}_{v}}, \overrightarrow{{\cal B}_{w}}, {\cal C}\}$, denote by $$\overrightarrow{{\cal B}} \leftarrow \{\overrightarrow{{\cal B}_{v}}, \overrightarrow{{\cal B}_{w}}, {\cal C}\}.$$

We have the following recurrence relation:
$$DP[u, \overrightarrow{{\cal B}}] = \min_{\overrightarrow{{\cal B}} \leftarrow \{\overrightarrow{{\cal B}_{v}}, \overrightarrow{{\cal B}_{w}}, {\cal C}\}} \left(DP[v, \overrightarrow{{\cal B}_{v}}] + DP[w, \overrightarrow{{\cal B}_{w}}] + c({\cal C})\right).$$

Although there are $2^q\times 2^q \times 4$ possible tuples
$\{\overrightarrow{{\cal B}_{v}}, \overrightarrow{{\cal B}_{w}}, {\cal C}\}$, less number of tuples need to be considered when computing $DP[u, \overrightarrow{{\cal B}}]$:

If ${\cal C} =\{uv,uw\}$ (i.e., both $uv$ and $uw$ are removed), then $\overrightarrow{{\cal B}} = \overrightarrow{0}$ regardless of $\overrightarrow{{\cal B}_{v}}$ and $\overrightarrow{{\cal B}_{w}}$. So we only need to consider one tuple $\{\overrightarrow{{\cal B}_{v}}^*, \overrightarrow{{\cal B}_{w}}^*, \{uv,uw\}\}$, where $DP[v, \overrightarrow{{\cal B}_{v}}^*]$ is minimum among all $DP[v, \overrightarrow{{\cal B}_{v}}]$ and $DP[w, \overrightarrow{{\cal B}_{w}}^*]$ is minimum among all $DP[w, \overrightarrow{{\cal B}_{w}}]$.

If ${\cal C} = \{uv\}$ (i.e., $uv$ is removed), then $\overrightarrow{{\cal B}} = \overrightarrow{{\cal B}_{w}}$ regardless of $\overrightarrow{{\cal B}_{v}}$. So we only need to consider $2^q$ tuples $\{\overrightarrow{{\cal B}_{v}}^*, \overrightarrow{{\cal B}_{w}}, \{uv\}\}$, where $DP[v, \overrightarrow{{\cal B}_{v}}^*]$ is minimum among all $DP[v, \overrightarrow{{\cal B}_{v}}]$. Similarly, if ${\cal C} = \{uw\}$, only $2^q$ tuples are considered.

If ${\cal C} = \emptyset$ (neither edge is removed), then the two vectors $\overrightarrow{{\cal B}_{v}}$ and $\overrightarrow{{\cal B}_{w}}$ cannot both have 1 at the same position $i$ because this would imply $T_u$ has two terminals in $S_i$. Therefore each position $i$ has only three possible pairs of values in the two vectors: (0,0) (1,0), and (0,1). In other words, only $3^q$ tuples are considered.

Therefore, for any vertex $u$, it takes time $O(3^q)$ to compute $DP[u, \overrightarrow{{\cal B}}]$ for all $\overrightarrow{{\cal B}}$.

{\bf Case 2. $u$ has only one child $v$.} Similar to the above, we have the following recurrence relation:
$$DP[u, \overrightarrow{{\cal B}}] = \min_{\overrightarrow{{\cal B}} \leftarrow \{\overrightarrow{{\cal B}_{v}}, {\cal C}\}} DP[v, \overrightarrow{{\cal B}_{v}}] + c({\cal C}).$$
In this case, it takes time $O(2^q)$ to compute $DP[u, \overrightarrow{{\cal B}}]$ for all $\overrightarrow{{\cal B}_u}$.

{\bf Case 3. $u$ is a leaf.} By our assumption, $u$ must be a terminal. We say that the vector $\overrightarrow{{\cal B}} = b_1b_2\ldots b_q$ is valid if $b_i = 1$ ($1 \leq i \leq q$) if and only if $u \in S_i$.
The formula for computing $DP[u, \overrightarrow{{\cal B}}]$ is
\begin{displaymath}
DP[u, \overrightarrow{{\cal B}}]
= \left\{
\begin{array}{ll}
0 & ~~~~~~~~\mbox{if $\overrightarrow{{\cal B}_u}$ is valid.}\\
+\infty & ~~~~~~~~\mbox{if $\overrightarrow{{\cal B}_u}$ is not valid.}
\end{array}
\right.
\end{displaymath}
In this case, it takes time $O(2^q)$ to compute $DP[u, \overrightarrow{{\cal B}}]$ for all $\overrightarrow{{\cal B}}$.

We can compute $DP[u, \overrightarrow{{\cal B}}]$ for all vertices $u$ in the tree $T$ in a bottom-up manner.
The optimal value $\mathrm{OPT}$ is equal to
$\min(DP[r, \overrightarrow{{\cal B}}])$. By backtracking, we can recover the optimal solution achieving
$\mathrm{OPT}$. Since the computation at any given vertex $u$ takes time $O(3^q)$, the total running time of our dynamic programming algorithm is
$O(3^q n)$, which is linear in $n$ because $q$ is a constant.
Thus we have
\begin{theorem}
\label{th - GMWCT can be solved in linear time}
The Generalized Multiway Cut problem on trees can be solved in linear time.
\end{theorem}

\section{Concluding remarks}
\label{sec:conclusion}
In this paper we gave a parameterized algorithm for the {\sc MCT} problem that runs in time $O^*(1.555^k)$, improving the previous algorithm by Chen et al. that runs in time $O^{*}(1.619^k)$. The aforementioned result implies that the {\sc GMWCT} problem is solvable in time $O^*(1.555^k)$ as well. We also proved that the special case of the {\sc GMWCT} problem in which the number of terminal sets is fixed is solvable in linear time, answering an open question by Liu and Zhang~\cite{LZ12}.

There are several questions related to the problems under consideration that remain open and are worth pursuing. First, it is interesting to seek further improvements on the running time of parameterized algorithms for the {\sc MCT} and the {\sc GMWCT} problems. The {\sc vertex cover} problem admits a parameterized algorithm that runs in time $O^*(1.2738^k)$~\cite{ckj}, and one can ask if the connection between {\sc MCT} and {\sc vertex cover} can be exploited further to improve the running time of parameterized algorithms for {\sc MCT} further. Another interesting research direction is related to kernelization. Currently, {\sc MCT} is know to admit a kernel of size $O(k^3)$~\cite{multicut}, and it is interesting to investigate if the problem admits a quadratic, or even a linear, kernel. The $O(k^3)$ kernel for {\sc MCT} relies heavily on kernelization techniques used for {\sc vertex cover}. So again, it would be interesting to investigate if one can exploit this connection further to obtain improvements on the kernel size for the {\sc MCT} problem. We leave those as open questions for future research.

\bibliographystyle{plain}
\bibliography{ref}

\end{document}